\documentclass[pdflatex,sn-mathphys-num]{sn-jnl}

\usepackage[utf8]{inputenc}

\usepackage{graphicx}%
\usepackage{multirow}%
\usepackage{amsmath,amsthm,amssymb,amsfonts}%
\usepackage{amsthm}%
\usepackage{mathrsfs}%
\usepackage[title]{appendix}%
\usepackage{xcolor}%
\usepackage{textcomp}%
\usepackage{manyfoot}%
\usepackage{booktabs}%
\usepackage{algorithm}%
\usepackage{algorithmicx}%
\usepackage{algpseudocode}%
\usepackage{listings}%
\usepackage{mdframed}
\usepackage{adjustbox}
\usepackage{etoolbox}
\usepackage{subcaption}

\allowdisplaybreaks[1] 

\theoremstyle{thmstyleone}%
\newtheorem{proposition}{Proposition}
\newtheorem{conjecture}{Conjecture}

\theoremstyle{thmstyletwo}%

\theoremstyle{thmstylethree}%

\begin{document}

\title[Investigation of Circular Photon Orbits in Naked Singularity Spacetimes from a Geometric Method]{Investigation of Circular Photon Orbits in Naked Singularity Spacetimes from a Geometric Method}


\author*[1]{\fnm{Qing-Rong} \sur{Xie}}\email{xieqingrong2026@163.com}

\author*[1,2]{\fnm{Chen-Kai} \sur{Qiao}}\email{chenkaiqiao@cqut.edu.cn}

\affil[1]{\orgdiv{School of Mathematical Sciences}, \orgname{Chongqing University of Technology}, 
	\orgaddress{\street{Hongguang Avenue}, 
		\city{Chongqing}, \postcode{400054}, 
		\country{China}}}

\affil[2]{\orgdiv{School of Physics and New Energy}, \orgname{Chongqing University of Technology}, 
\orgaddress{\street{Hongguang Avenue}, 
\city{Chongqing}, \postcode{400054}, 
\country{China}}}


\abstract
{
Circular photon orbits play a pivotal role in both gravitational theories and astronomical observations. However, the properties of circular photon orbits in naked singularity spacetimes remain insufficiently explored and deserve in-depth investigation. The present work is dedicated to a comprehensive study of the features of circular photon orbits in naked singularity spacetimes. Notably, a geometric approach is employed to investigate these orbits, in which the framework of optical geometry together with its intrinsic curvatures plays a crucial role. We analyze the existence of circular photon orbits through the intrinsic geodesic curvature, and then examine the numbers of stable and unstable circular orbits, as well as the topological invariant associated with circular photon orbits. By examining various classes of naked singularity spacetimes, we obtain a general conclusion regarding circular photon orbits that holds for the most of spherically symmetric naked singularity spacetimes: the total number of circular photon orbits is an even integer ($\textit{N} = \text{2}\textit{k}$), consisting of $\textit{k}$ stable orbits and $\textit{k}$ unstable orbits in equal proportion. A mathematical proof of this conclusion is also provided in the present work. Furthermore, a comparison with other categories of spacetimes reveals that the conclusions regarding the number of circular photon orbits in naked singularity spacetimes agree with those obtained for compact object spacetimes without naked singularities, indicating that the number of circular photon orbits is primarily governed by the presence of event horizons rather than spacetime singularities.
}




\maketitle

\section{Introduction}\label{sec Introduction}

Circular photon orbits in gravitational fields are of great significance for theoretical physics, mathematical physics, and astronomical observations. From a theoretical perspective, circular photon orbits are related to the topology and causal structure of spacetime \cite{Cunha_2017,Cunha_2020,Wei_2020}, as well as spacetime stability \cite{Cunha_2023a}. Investigations of circular photon orbits contribute substantially to our understanding of the mathematical properties of spacetime. From an observational astronomy standpoint, circular photon orbits are intimately connected to measurable phenomena including black hole shadow imaging \cite{EHT2019,EHTb2019,EHT2022,Gralla_2019,Perlick_2022,Vagnozzi_2023,chen2023black,Wei_2013}, gravitational lensing \cite{Virbhadra_2000,Bozza_2002,Iyer_2007,Tsukamoto_2017}, and quasi-normal modes in gravitational radiation \cite{Cardoso_2009,Cardoso_2014,Yang_2012,Li_2021}.

There are many approaches to investigating circular photon orbits in gravitational fields. The standard approach relies on the effective potential of photons in curved spacetime. Owing to its long history of development, this framework is extensively adopted in textbook treatments \cite{hartle2021gravity}. Beyond this standard scheme, progress in mathematical physics has led to the development of a wealth of innovative analytical techniques for the study of circular photon orbits (as well as circular orbits of other particles), such as topological methods \cite{Cunha_2017,Cunha_2020,Cunha_2018,Wei_2020,Wei_2023,chen2026universal} and Lyapunov exponent analysis \cite{deich2024lyapunov,Cardoso_2009}. These techniques have enabled the derivation of several important theorems and conclusions concerning circular photon orbits, greatly advancing our understanding of such orbits. Notably, in recent years, the second author of this paper proposed a geometric approach, from which the circular photon orbits and their stability are determined via the geodesic curvature and Gaussian curvature of optical geometry \cite{Qiao_2022a,Qiao_2022b} \footnote{In the case of rotating spacetimes, the flag curvature is additionally involved in determining these circular photon orbits \cite{Qiao_2025c}.}. This method establishes a connection between circular photon orbits and the geometric properties of optical geometry, offering a novel perspective for understanding circular photon orbits from intrinsic curvatures. To date, this approach has been applied to a wide variety of gravitational systems \cite{Qiao_2025a,Qiao_2025b,Qiao_2025c,cunha2022null,Bermdez_Crdenas_2025,Bermdez_Crdenas_2025b,Gallo_2025,zhang2026gaussian,andino2026perturbative,zhang2026geometric}.

Recently, the properties of circular photon orbits in black hole spacetimes, astrophysical compact object spacetimes, and other spacetimes in gravity theories have been explored extensively \cite{cunha2017fundamental,Claudel_2001,Guo_2023,Cveti_2016,xavier2024traversable,Hod_2013,Mishra_2019,song2025existence,song2026bounds,li2025images,diazguerra2026photon,Vertogradov_2024_photonsphere,Vertogradov_2024_b,Ara_jo_Filho_2024,Heidari_Filho_2025}. Specifically, the existence of circular photon orbits, the total number of circular photon orbits, the number and distribution patterns of stable and unstable circular orbits, and the topological invariant of spacetime associated with circular photon orbits have been reported in the literature \cite{cederbaum2017,cederbaum2016,Jia_2018a,Jia_2018b,Guo_2021,Ghosh_2021,Cunha_2020,cunha2022null,Wei_2020,Qiao_2025a,Cunha_2017,Qiao_2025b}. However, in contrast to the black hole spacetimes and ultra-compact objects spacetimes, the universal features and properties of circular photon orbits in naked singularity spacetimes have not yet been clearly established. Naked singularity spacetimes constitute one of the most exotic classes of spacetimes in general relativity and other alternative gravity theories, which may have profound influences on photon orbits and causal structures \cite{christodoulou1994,joshi2020shadow,ziaie2011naked,JOSHI_2011,virbhadra2002,Wang_MingZhi_2024,Chakraborty_2019,Chakraborty_2017}. The profound influence exerted by naked singularities on circular photon orbits is still poorly understood. It is therefore highly necessary to investigate the behaviors and characteristics of circular photon orbits in naked singularity spacetimes, which constitutes the main purpose of the present work. 

In this paper, we study the existence of circular photon orbits, the total number of circular photon orbits, and the count of stable and unstable circular orbits in general naked singularity spacetimes (with spherical symmetry) \footnote{For simplicity, we assume that the spacetime admits only a finite number of circular photon orbits, thereby excluding the ``pathological'' scenario of infinitely many circular orbits or a continuum of circular photon orbits (in the latter case, the circular orbits are not discrete but form a continuous family of orbits filling a radial interval $r \in [r_a, r_b]$). We further assume that every circular photon orbit is either stable or unstable, thus excluding degenerate circular photon orbits (those with $\frac{d^2V}{dr^2}=0$) for which the concept of stability becomes ill-defined.}. Particularly, we focus on exploring how the divergent behaviors of the metric and its derivatives at naked singularities affect circular photon orbits. Concretely, we first select several representative classes of naked singularity metrics to examine the existence and number of circular photon orbits from geodesic curvature and Gaussian curvature. By comparing these results, some common features of the circular photon orbits in naked singularity spacetimes are identified, and rigorous mathematical proofs for these conclusions are further provided in our work. These conclusions obtained in this work impose essential constraints on circular photon orbits in naked singularity spacetimes. 

This paper is structured as follows. Section \ref{sec Introduction} elaborates on the research background and motivation of this work. Section \ref{sec Geometric Approach} presents the geometric approach for investigating circular photon orbits. Section \ref{sec PO Naked Singularity} explores the properties of circular photon orbits in naked singularity spacetimes. Section \ref{sec Comparison} gives a comparison of circular photon orbits in naked singularity spacetimes and those in other categories of spacetimes. Section \ref{sec Conclusion} provides the conclusions of this work and outlines prospects for future research.

\section{Geometric Approach to Circular Photon Orbits}\label{sec Geometric Approach}

This section reviews the geometric method adopted to analyze the circular photon orbits. This approach, which was first proposed by the second author of this work, has been applied to a number of studies in recent years \cite{Qiao_2022a,Qiao_2022b,Qiao_2025a,Qiao_2025b,Qiao_2025c,cunha2022null}. The main idea of this approach is to investigate the existence and distribution of circular photon orbits by means of intrinsic geometric curvatures in low-dimensional geometry. In particular, the optical geometry, which can be motivated by the generalized Fermat's principle \cite{Ye_2008_optical,Piwnik_2025_Fermat}, is convenient for realizing such a geometric method, and it serves as a low-dimensional reduction of the 4-dimensional Lorentzian spacetime. 

In mathematics, optical geometry can be constructed in several ways \cite{Gibbons_2008,Gibbons_2009,Gibbons_2009b}. One intuitive approach is to impose the null constraint $d\tau^{2}=-ds^{2}=0$ under continuous mappings of spacetime geometry
\begin{equation}
	\underbrace{ds^{2} = g_{\mu\nu}dx^{\mu}dx^{\nu}}_{\text{Spacetime Geometry}}
	\ \ \overset{d\tau^{2}=-ds^{2}=0}{\Longrightarrow} \ \ 
	\underbrace{dt^{2} = g^{\text{OP}}_{ij}dx^{i}dx^{j}}_{\text{Optical Geometry}}
	\label{optical geometry} \ \ \text{or} \ \   \underbrace{dt=\sqrt{\alpha^{\text{OP}}_{ij}dx^{i}dx^{j}}+\beta^{\text{OP}}_{i}dx^{i}}_{\text{Optical Geometry}}
\end{equation}
The features of optical geometry are determined by the symmetry of Lorentzian spacetime. In particular, for a spherically symmetric spacetime, the corresponding optical geometry gives rise to a Riemannian geometry, where $dt^{2}=g_{ij}dx^{i}dx^{j}$. On the other hand,  for a stationary axially symmetric spacetime, the corresponding optical geometry becomes the Randers-Finsler geometry with $dt=\sqrt{\alpha_{ij}dx^{i}dx^{j}}+\beta_{i}dx^{i}$. Optical geometry plays an important role in the research of gravitational deflection and gravitational lensing \cite{Gibbons_2008,Werner_2012,Ishihara_2016,Ishihara_2017,Ono_2017,Ono_2019,Li_2020_gauss,HuangYang_2024,basumallick2026weak,Ovgun_2026_gauss}. Particularly, in the spherically symmetric spacetime, the evolution of particle orbits can always be confined in the equatorial plane, which is a property that is universally held under the influences of spherical symmetry. Accordingly, a 2-dimensional optical geometry can be constructed for such cases.
\begin{equation}
	\underbrace{dt^{2} = g^{\text{OP}}_{ij}dx^{i}dx^{j}}_{\text{Optical Geometry}}
	\ \ \overset{\theta=\pi/2}{\Longrightarrow} \ \ 
	\underbrace{dt^{2}=\tilde{g}^{\text{OP-2d}}_{ij}dx^{i}dx^{j}}_{\text{Optical Geometry (Two Dimensional)}}
	\label{eq:6}
\end{equation}

When adopting this geometric method, it is found that the existence and stability of circular photon orbits can be analyzed via the geodesic curvature and Gaussian curvature in optical geometry, respectively. Both of them are intrinsic geometric quantities that are independent of local coordinate transformations. These intrinsic curvatures serve as crucial tools for determining photon orbits and their stability. First, by mapping circular photon orbits into optical geometry, these orbits maintain the properties of geodesics. Consequently, their geodesic curvature equals zero in this case \cite{Qiao_2022a,Qiao_2022b,Qiao_2025a,Qiao_2025b,Qiao_2025c}
\begin{equation} 
	\kappa_{g}(r=r_{ph})  =  0 . \label{eq:7}
\end{equation}
Secondly, the stability of a circular orbit can be determined by analyzing conjugate points in optical geometry. For stable and unstable circular photon orbits, they exhibit entirely different behaviors when undergoing slight perturbations. For a stable circular photon orbit, the perturbed orbits may form new bounded trajectories in its vicinity. In contrast, the orbits perturbed from an unstable circular photon orbit will drift far away. From the mathematical concept of conjugate points, we arrive at the vital conclusion: there are conjugate points in the stable circular orbit, while no conjugate points exist in the unstable circular orbit. The presence and absence of conjugate points provide us with a novel scheme to distinguish the stable and unstable circular photon orbits. Specifically, it is important to emphasize that the Cartan-Hadamard theorem in differential geometry and topology strongly constrains the Gaussian curvature and the existence of conjugate points \cite{berger2012differential,do2016differential}. Applying the Cartan-Hadamard theorem in optical geometry, the following Gaussian curvature condition on the stability of circular photon orbits has been derived \cite{Qiao_2022a,Qiao_2022b}.
\begin{eqnarray}
	\mathcal{K}(r) < 0 & \Rightarrow & \text{The circular photon orbits $r=r_{ph}$ is unstable}  \nonumber
	\\
	\mathcal{K}(r) > 0 & \Rightarrow & \text{The circular photon orbits $r=r_{ph}$ is stable}  \label{gauss}
\end{eqnarray}
Since the geodesic and Gaussian curvature conditions in  (\ref{eq:7}) and (\ref{gauss}) for circular photon orbits are derived purely from the geometric and topological properties of the optical geometry, which is independent of the specific metric forms of the particular gravitational systems, these conclusions can be applied to any spherically symmetric system (no matter whether the naked singularity is existed). When analyzing the properties of circular photon orbits, the geometric method based on intrinsic curvatures yields equivalent results to the conventional effective potential method \cite{Qiao_2022a,Qiao_2022b}
\begin{subequations}
\begin{eqnarray} 
	\kappa_{g}(r=r_{\text{ph}})  =  0  
	\ \ & \Leftrightarrow & \ \
	\frac{dV_{\text{eff}}(r)}{dr} \bigg|_{r=r_{\text{ph}}} = 0 ,
	\\
	\mathcal{K}(r=r_{\text{unstable}}) < 0  
	\ \ & \Leftrightarrow & \ \
	\frac{d^{2}V_{\text{eff}}(r)}{dr^{2}} \bigg|_{r=r_{\text{unstable}}} < 0 ,
	\\ 
	\mathcal{K}(r=r_{\text{stable}}) > 0
	\ \ & \Leftrightarrow & \ \
	\frac{d^{2}V_{\text{eff}}(r)}{dr^{2}} \bigg|_{r=r_{\text{stable}}} > 0 .
\end{eqnarray}
\end{subequations}

In this work, the static and spherically symmetric naked singularity spacetimes are focused. The general metric form of such kind of spacetimes can be expressed as
\begin{equation}
	ds^{2} = g_{tt}(r) dt^{2} + g_{rr}(r) dr^{2} 
	+ g_{\theta\theta}(r) d\theta^{2} 
	+ g_{\phi\phi}(r,\theta) d\phi^{2} .
	\label{eq:9}
\end{equation}
In addition, because of the spherical symmetry, the metric components $g_{\theta\theta}=r^{2}$ and $g_{\phi\phi}=r^{2}\sin^{2}\theta$ can be obtained via coordinate transformations. Accordingly, the spacetime metric can be simplified to the following form
\begin{equation}
	ds^{2} = - f(r) dt^{2} + g(r) dr^{2} 
	+ r^{2} d\theta^{2} + r^{2} \sin^{2}\theta d\phi^{2} ,
	\label{eq:10}
\end{equation}
with $f(r)=-g_{tt}(r)$ and $g(r)=g_{rr}(r)$. The corresponding optical metric of spherically symmetric spacetimes, when restricted to the equatorial plane ($\theta=\pi/2$), can be constructed as
\begin{equation}
	dt^{2} = \tilde{g}^{\text{OP-2d}}_{ij}dx^{i}dx^{j} 
	= \frac{g(r)}{f(r)} \cdot dr^{2} + \frac{r^{2}}{f(r)} \cdot d\phi^{2} . \label{eq:11}
\end{equation} 
Using the 2-dimensional optical metric, the geodesic curvature of circular photon orbits with constant radius $r$ can be calculated using the Liouville formula \cite{do2016differential}
\begin{equation}
	\kappa_{g}(r) 
	= \frac{1}{2\sqrt{\tilde{g}^{\text{OP-2d}}_{rr}}} 
	\frac{\partial \big[ \text{log}(\tilde{g}^{\text{OP-2d}}_{\phi\phi})\big]}{\partial r}
	= \frac{1}{\sqrt{f(r) \cdot g(r)}} 
	\bigg[ \frac{f(r)}{r} - \frac{1}{2} \frac{df(r)}{dr} \bigg] .
	\label{eq:12}
\end{equation}
Similarly, the Gaussian curvature of the 2-dimensional equatorial plane at the circular photon orbits position can be calculated by the following formula
\begin{eqnarray}
	\mathcal{K} 
	& = & -\frac{1}{\sqrt{\tilde{g}^{\text{OP-2d}}}} 
	\bigg[
	\frac{\partial}{\partial \phi} \bigg( \frac{1}{\sqrt{\tilde{g}^{\text{OP-2d}}_{\phi\phi}}} \frac{\partial\sqrt{\tilde{g}^{\text{OP-2d}}_{rr}}}{\partial \phi}  \bigg)
	+ \frac{\partial}{\partial r} \bigg( \frac{1}{\sqrt{\tilde{g}^{\text{OP-2d}}_{rr}}} \frac{\partial\sqrt{\tilde{g}^{\text{OP-2d}}_{\phi\phi}}}{\partial r}  \bigg)
	\bigg]  \nonumber
	\\
	& = &   \frac{1}{2r} \frac{1}{g(r)} \frac{df(r)}{dr} 
	+ \frac{1}{2r} \frac{f(r)}{[g(r)]^2} \frac{dg(r)}{dr} 
	- \frac{1}{2f(r) g(r)} \bigg[ \frac{df(r)}{dr} \bigg]^2
	- \frac{1}{4 [g(r)]^2} \frac{df(r)}{dr} \frac{dg(r)}{dr} 
	+ \frac{1}{2g(r)} \frac{d^2 f(r)}{dr^2} , 
	\ \ \ \ \ \ 
	\label{eq:13}
\end{eqnarray}
with $\tilde{g}^{\text{OP-2d}}=\text{det} \big( \tilde{g}^{\text{OP-2d}}_{ij} \big)$ being the determinant of the 2-dimensional optical geometry metric. In this work, the geodesic curvature and Gaussian curvature in equations \eqref{eq:12}--\eqref{eq:13} are fundamental geometric quantities in the exploration of circular photon orbits in naked singularity spacetimes. 

\section{Features of Circular Photon Orbits in Naked Singularity Spacetimes}\label{sec PO Naked Singularity}

In this section, we give an investigation of the features of circular photon orbits in nakedly singular spacetimes, including the existence of circular photon orbits, the number of stable and unstable circular orbits, as well as the topological charge associated with circular photon orbits. 

From the perspective of our geometric method, the circular photon orbits are closely related to geodesic curvature via $\kappa_{g}=0$. Accordingly, such circular orbits exist if and only if the equation $\kappa_{g}(r) = 0$ admits at least one solution. In this way, to explore the existence of circular photon orbits in spacetimes with naked singularities, we need to investigate the variation trend of geodesic curvature $\kappa_{g}(r)$ as a function of the radial coordinate $r$, especially its behavior at the naked singularity $r \to 0$ \footnote{Since we focus on the spherically symmetric naked singularity spacetimes, the spacetime singularity should be confined to the center $r=0$.} and at the outer boundary of the considered spacetime region (as illustrated in figure \ref{fig:A}). 

\begin{figure*}
	\centering
	\includegraphics[width=0.65\textwidth]{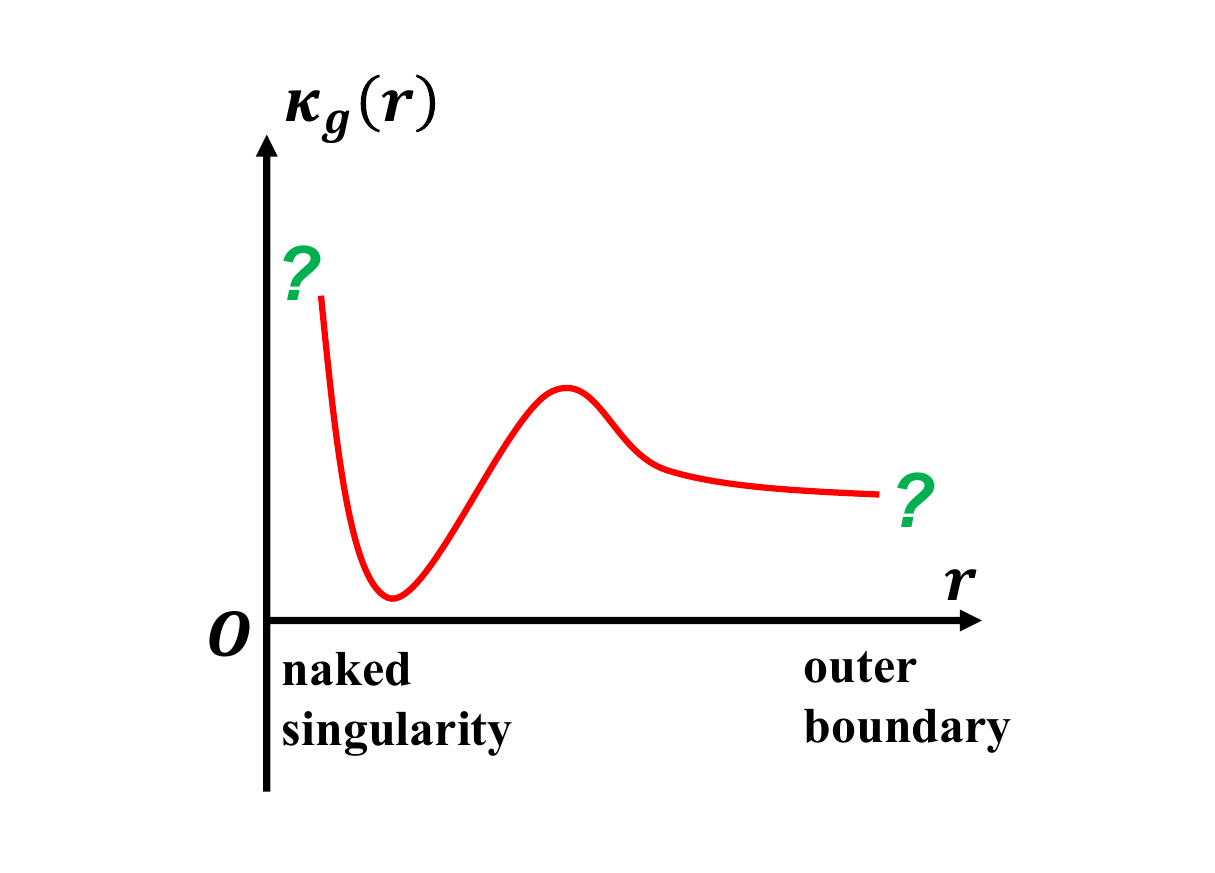}
	\caption{Illustration of the behavior of geodesic curvature $\kappa_{g}(r)$, which is used to determine the number of solutions to the geodesic constraint equation $\kappa_{g}(r) = 0$.}
	\label{fig:A}
\end{figure*}

In particular, the spacetime becomes singular at the central naked singularity, where the Riemannian curvature or other curvature invariants may become divergent (inducing the second derivatives of the metric to diverge at this point). However, the naked singularity does not impose strict constraints on the behavior of the metric functions $f(r)$, $g(r)$, or their first derivatives. In the current work, we generally assume that $f(r)>0$ and $g(r)>0$ hold everywhere for the considered spacetime region, such that the signature of Lorentzian spacetime is kept and no horizons emerge in the neighborhood of the naked singularity. From the geodesic curvature in equation (\ref{eq:12}), it turns out that the limiting behavior of the first derivative $\frac{df(r)}{dr}$ as $r \to 0$ has a significant influence on the geodesic curvature near the naked singularity. We can classify naked singularity spacetimes into three classes according to the limiting behavior of $\frac{df(r)}{dr}$: $\lim_{r \to 0}\frac{df(r)}{dr}= \text{finite}$, $\lim_{r\to 0}\frac{df(r)}{dr}=-\infty$, $\lim_{r\to 0}\frac{df(r)}{dr}=+\infty$. The investigations of geodesic curvature behaviors and circular photon orbits for these classes of naked singularity spacetimes are provided in subsections \ref{sec3.1} and \ref{sec3.2}, respectively. Notably, in this section we consider spacetimes with $f(r)>0$ as $r\to 0$, while the case $f(r)\to 0$ is treated in Appendix \ref{appendix C}. Furthermore, throughout the analysis, we also exclude the situation in which $\frac{df(r)}{dr}$ oscillates infinitely often as $r\to 0$ (in which case the corresponding limit fails to exist), since it would lead to a pathological, ``catastrophic'' scenario with infinitely many circular photon orbits.

Meanwhile, the geodesic curvature at the outer boundary of the considered spacetime region is not affected by the naked singularity, and it depends solely on the asymptotic behavior of spacetimes. In this work, we consider the most common asymptotic behaviors in general relativity and other gravitational theories, namely asymptotically flat, asymptotically de Sitter, and asymptotically anti-de Sitter spacetimes \footnote{For the asymptotically flat and asymptotically anti-de Sitter spacetimes, the outer boundary is located at the spatial infinity $r \to \infty$. In contrast, asymptotically de Sitter spacetimes possess a cosmological horizon $r_{\text{cosmic}}$, so the spacetime region we consider is always inside the cosmological horizon. Hence, the outer boundary of asymptotically de Sitter spacetimes is at the cosmological horizon $r \to r_{\text{cosmic}}$.}. The geodesic curvature at the outer boundary of these spacetimes is discussed in Appendix \ref{appendix A}. 

\subsection{Naked singularity spacetimes with $\lim_{r\to 0}\frac{df(r)}{dr} =\text{finite}$ or $\lim_{r\to 0}\frac{df(r)}{dr}=-\infty$}\label{sec3.1}

In this subsection, we consider two cases in which the metric derivative gets a finite value or diverges to negative infinity at the naked singularity  ($\lim_{r\to 0}\frac{df(r)}{dr} =\text{finite}$ or $\lim_{r\to 0}\frac{df(r)}{dr}=-\infty$). In such cases, the geodesic curvature $\kappa_{g}(r)$ in the central limit becomes
\begin{eqnarray}
 	\lim_{r\to 0}\frac{df(r)}{dr} = \text{finite}  \nonumber
 	\ \ \Rightarrow \ \ 
 	\lim_{r\to 0}\kappa_{g}(r) &=& \lim_{r\to 0}\left[\frac{1}{\sqrt{f(r)\cdot g(r)}}\left(\frac{f(r)}{r}-\frac{1}{2}\frac{df(r)}{dr}\right)\right]\notag
 	\\
 	&=& \lim_{r\to 0}\left[\frac{1}{\sqrt{f(r)\cdot g(r)}}\left(+\infty-\text{finite}\right)\right] \nonumber
 	\\
 	&=& +\infty	,
    \\
 	\lim_{r\to 0}\frac{df(r)}{dr} = -\infty  \nonumber
 	\ \ \Rightarrow \ \ 
 	\lim_{r\to 0}\kappa_{g}(r) &=& \lim_{r\to 0}\left[\frac{1}{\sqrt{f(r)\cdot g(r)}}\left(\frac{f(r)}{r}-\frac{1}{2}\frac{df(r)}{dr}\right)\right]\notag
 	\\
 	&=& \lim_{r\to 0}\left[\frac{1}{\sqrt{f(r)\cdot g(r)}}\left(+\infty-(-\infty)\right)\right] \nonumber
 	\\
 	&=& +\infty	.
\end{eqnarray}
It should be noted that in a naked singularity spacetime, the metric signature is preserved (namely $f(r)>0$ and $g(r)>0$) to prevent the emergence of event horizons enclosing the singularity, thereby ensuring that the $f(r) \cdot g(r)$ does not vanish at the naked singularity $r=0$ \footnote{Furthermore, we also require that $f(r) \cdot g(r)$ does not diverge to the positive infinity such that the spacetime volume $dV = \sqrt{|\text{det}(g_{\mu\nu})|} d^4 x$ remains well-defined and non-divergent as $r$ approaches the naked singularity $r \to 0$}. Provided that the geodesic curvature in the central limit tends to positive infinity ($\lim_{r\to 0}\kappa_{g}(r)=+\infty$), combined with the condition that the geodesic curvature satisfies $\kappa_{g}(r)\to 0^+$ or $\kappa_{g}(r)>0$ at the outer boundary (see Appendix \ref{appendix A}), the constraint equation for geodesic curvature $\kappa_{g}(r)=0$ either admits no solutions or possesses an even number of solutions under this circumstance. Consequently, for naked singularity spacetimes with $\lim_{r\to 0}\frac{df(r)}{dr}=\text{finite}$ or $\lim_{r \to 0}\frac{df(r)}{dr}=-\infty$, the existence of circular photon orbits is determined, and the number of circular photon orbits is $N=2k$ (where $k \in \mathbb{N}$ is a natural integer). Furthermore, from the alternating distribution property of stable and unstable circular photon orbits (which is introduced in Appendix \ref{appendix B}), the numbers of stable and unstable circular photon orbits are both equal to $k$. The topological invariant associated with circular photon orbits becomes $w = n_{\text{stable}} - n_{\text{unstable}} = 0$ \footnote{According to the topological method on circular photon orbits in reference \cite{Wei_2020,Cunha_2020}, each unstable circular photon orbit (also called the standard circular photon orbit) can be assigned a topological charge of $w=-1$, while each stable circular photon orbit (also called the exotic circular photon orbit) can be assigned a topological charge of $w=1$. The sum of topological charge over all circular photon orbits yields a topological invariant of the spacetime. Therefore, given the number of stable and unstable circular photon orbits, the topological charge / topological invariant associated with circular photon orbits in this spacetime is $w = n_{\text{stable}} - n_{\text{unstable}}$.}.

\subsection{Naked Singularity spacetimes with $\lim_{r\to 0}\frac{df(r)}{dr}=+\infty$}\label{sec3.2}

In this subsection, we provide a discussion regarding the metric derivative diverging to positive infinity as $r$ approaches to the central naked singularity ($\lim_{r\to 0}\frac{df(r)}{dr}=+\infty$). We assume that $f(r)>0, g(r)>0$ such that $f(r) \cdot g(r)$ does not vanish throughout the considered spacetime region. This assumption ensures that the Lorentzian signature of spacetime is kept, and no coordinate singularities (which may bring about event horizons or infinite redshift surfaces) exist in the spacetime. It can be clearly seen that the divergent behavior of the metric derivative $\frac{df(r)}{dr}$ exerts dramatic influences on the computation of geodesic curvature $\kappa_g(r)$ at the naked singularity, which makes it slightly different from those presented in subsection \ref{sec3.1}. Unfortunately, the direct evaluation of the geodesic curvature leads to an indeterminate result in the limit of $r \to 0$
\begin{eqnarray}
	\lim_{r\to 0}\frac{df(r)}{dr} = +\infty  \nonumber
	\ \ \Rightarrow \ \ 
	\lim_{r\to 0}\kappa_{g}(r) &=& \lim_{r\to 0}\left[\frac{1}{\sqrt{f(r)\cdot g(r)}}\left(\frac{f(r)}{r}-\frac{1}{2}\frac{df(r)}{dr}\right)\right]\notag
	\\
	&=& \lim_{r\to 0}\left[\frac{1}{\sqrt{f(r)\cdot g(r)}}\left(+\infty-(+\infty)\right)\right] \nonumber
	\\
	&=& \text{indefinite}		
\end{eqnarray} 
Therefore, to acquire definitive conclusions, it is necessary to specify typical analytical functions $f(r)$ and discuss whether the geodesic curvature constraint equation $\kappa_g(r)=0$ possesses solutions (i.e., whether circular photon orbits exist) in such cases. In the following part, we begin with several typical metric functions $f(r)$, which frequently appear in gravity theories, to derive results concerning circular photon orbits, and then summarize the common features and conclusions shared by all cases. Based on these commonalities, we finally propose a conjecture on circular photon orbits that holds universally for naked singularity spacetime with $\lim_{r\to 0}\frac{df(r)}{dr} = +\infty$.

\textbf{Case I: Spacetime with $f(r)=ar^{\alpha}+f_{\text{regular}}$ behavior near the naked singularity.} When the metric function behaves as $f(r)=ar^{\alpha}+f_{\text{regular}}$ (where $a>0$, $0<\alpha<1$, $f_{\text{regular}}$ is a non-divergent regular term satisfying $\frac{df_{\text{regular}}}{dr} = \text{finite}$). In the naked singularity spacetimes without the event horizons, the metric must satisfy $f(r)>0$ in the entire spacetime region under consideration, which imposes a constraint on the regular part $f_{\text{regular}}$ in the central limit
\begin{equation}
	\lim_{r\to 0}f(r) 
	=\lim_{r\to 0}\left[ar^{\alpha}+f_{\text{regular}}\right] >0 
	\ \ \Rightarrow \ \ 
	\lim_{r\to 0}f_{\text{regular}} > 0 .
\end{equation}
In addition, the metric derivative diverges to positive infinity at the spacetime singularity $r \to 0$
\begin{equation}
	\lim_{r\to 0}\frac{df(r)}{dr}
	=\lim_{r\to 0}\left[a\alpha r^{\alpha-1}+\frac{df_{\text{regular}}}{dr}\right]
	=+\infty+\text{finite}
	=+\infty .
\end{equation}
which exactly meets the requirement for the class of naked singularity spacetimes considered in this subsection ($\lim_{r\to 0}\frac{df(r)}{dr}=+\infty$). Since we have assumed that $f(r) \cdot g(r)$ is non-vanishing, it is possible to take $h_0=\lim_{r\to 0}[g(r)\cdot f(r)]$ as the product of metric functions at the naked singularity. Hence, the geodesic curvature in the central limit becomes
\begin{eqnarray}
	\lim_{r\to 0}\kappa_{g}(r)
	&=& \lim_{r\to 0}\left[\frac{1}{\sqrt{f(r)\cdot g(r)}}\left(\frac{f(r)}{r}-\frac{1}{2}\frac{df(r)}{dr}\right)\right] \nonumber \\
	&=& \lim_{r\to 0}\left[ \frac{1}{\sqrt{h_0}} \cdot \left(  \frac{ar^{\alpha}+f_{\text{regular}}}{r} -\frac{a\alpha r^{\alpha-1}}{2} -\frac{1}{2}\frac{df_{\text{regular}}}{dr} \right) \right] \nonumber \\
	&=& \frac{1}{\sqrt{h_{0}}}\cdot\lim_{r\to 0}\frac{2a-a\alpha}{2r^{1-\alpha}} +\frac{1}{\sqrt{h_0}}\cdot\lim_{r\to 0}\frac{f_{\text{regular}}}{r} -\frac{1}{2\sqrt{h_0}}\cdot\lim_{r\to 0}\frac{df_{\text{regular}}}{dr} \nonumber \\
	&=& +\infty +\infty - \text{finite} \nonumber \\
	&=& +\infty		
\end{eqnarray}
Therefore, we can draw the following conclusion. In naked singularity spacetimes with $f(r)=ar^{\alpha}+f_{\text{regular}}$ behavior, the geodesic curvature satisfies $\lim_{r\to 0}\kappa_{g}(r)=+\infty$ at the central limit and $\kappa_{g}(r) \to 0^+$ (or $\kappa_{g}(r) > 0$) at the outer boundary (see Appendix \ref{appendix A}), indicating that the equation $\kappa_g(r) = 0$ admits either no solution or an even number of solutions, corresponding to the appearance of an even number of circular photon orbits $N=2k$ ($k \in \mathbb{N}$). Furthermore, from the alternating distribution property of stable and unstable circular photon orbits (which is presented in Appendix \ref{appendix B}), the numbers of stable and unstable circular photon orbits are both equal to $k$ in this case. The topological invariant associated with circular photon orbits becomes $w=n_{\text{stable}} - n_{\text{unstable}} = 0$.

\textbf{Case II: Spacetimes with $f(r)=\frac{a}{(\ln r)^{2n}}+f_{\text{regular}}$ and $f(r)=\frac{a}{(\ln r)^{2n+1}}+f_{\text{regular}}$ behavior near the naked singularity.} This scenario encompasses two analogous metric behaviors in the vicinity of naked singularity: $f(r)=\frac{a}{(\ln r)^{2n}}+f_{\text{regular}}$ and $f(r)=\frac{a}{(\ln r)^{2n+1}}+f_{\text{regular}}$. We first consider the metric $f(r)=\frac{a}{(\ln r)^{2n}}+f_{\text{regular}}$, with $a>0$, $n\in \mathbb{N}^+$ and $f_{\text{regular}}$ to be a non-divergent function. The assumption that $f(r)>0$ is kept in the spacetime region makes the regular part $f_{\text{regular}}$ to be positive at the naked singularity 
\begin{equation}
	\lim_{r\to 0}f(r)
	=\left[\frac{a}{(\ln r)^{2n}}+f_{\text{regular}}\right]>0
	\ \ \Rightarrow \ \ 
	\lim_{r\to 0}f_{\text{regular}} > 0 .
\end{equation}
Meanwhile, the first-order metric derivative diverges to positive infinity in the central limit 
\begin{equation}
	\lim_{r\to 0}\frac{df(r)}{dr}
	=\left[\lim_{r\to 0}\frac{-2an}{r(\ln r)^{2n+1}}+\frac{df_{\text{regular}}}{dr}\right]
	=+\infty+\text{finite}
	=+\infty.
\end{equation}
Let $h_0=\lim_{r\to 0}[g(r) \cdot f(r)]$ denotes the product of metric functions at the naked singularity, the geodesic curvature $\kappa_{g}(r)$ in the central limit becomes
\begin{eqnarray}
	\lim_{r\to 0}\kappa_{g}(r)
	&=& \lim_{r\to 0}\left[\frac{1}{\sqrt{f(r)\cdot g(r)}}\left(\frac{f(r)}{r}-\frac{1}{2}\frac{df(r)}{dr}\right)\right]\notag \nonumber \\
	&=& \lim_{r\to 0} \left\{ \frac{1}{\sqrt{h_0}} \left[\frac{\frac{a}{(\ln r)^{2n}}+f_{\text{regular}}}{r}+\frac{an}{r(\ln r)^{2n+1}}-\frac{1}{2}\frac{df_{\text{regular}}}{dr} \right] \right\}\notag \nonumber \\
	&=& \frac{1}{\sqrt{h_0}}\cdot\lim_{r\to 0}\frac{a(n+\ln r)}{r(\ln r)^{2n+1}} 
	+\frac{1}{\sqrt{h_0}}\cdot\lim_{r\to 0}\frac{f_{\text{regular}}}{r} 
	-\frac{1}{2\sqrt{h_0}}\cdot\lim_{r\to 0}\frac{df_{\text{regular}}}{dr} \nonumber \\
	&=& +\infty	+\infty -\text{finite}	\nonumber \\
	&=& +\infty	.			
\end{eqnarray}

Similarly, suppose the metric behaves as $f(r)=\frac{a}{(\ln r)^{2n+1}}+f_{\text{regular}}$ near the naked singularity. In this case, the coefficient differs from that in previous case, where it is required that $a<0$ and $n\in \mathbb{N}^+$ to ensure $f(r)>0$. This leads the regular part $f_{\text{regular}}$ to be positive in the naked singularity 
\begin{equation}
	\lim_{r\to 0}f(r)
	=\left[\frac{a}{(\ln r)^{2n+1}}+f_{\text{regular}}\right]>0
	\ \ \Rightarrow \ \ 
	\lim_{r\to 0}f_{\text{regular}} > 0 .
\end{equation}
Meanwhile, the first-order metric derivative diverges to positive infinity at the central limit 
\begin{equation}
	\lim_{r\to 0}\frac{df(r)}{dr}
	=\lim_{r\to 0}\left[\frac{-(2n+1)a}{r(\ln r)^{2n+2}}+\frac{df_{\text{regular}}}{dr}\right]
	=+\infty+\text{finite}
	=+\infty.		
\end{equation}
Under this circumstance, the geodesic curvature $\kappa_{g}(r)$ at the central limit becomes
\begin{eqnarray}
	\lim_{r\to 0}\kappa_{g}(r)
	&=& \lim_{r\to 0}\left[\frac{1}{\sqrt{f(r)\cdot g(r)}}\left(\frac{f(r)}{r}-\frac{1}{2}\frac{df(r)}{dr}\right)\right] \nonumber \\
	&=& \lim_{r\to 0} \left\{\frac{1}{\sqrt{h_0}}\left[\frac{\frac{a}{(\ln r)^{2n+1}}+f_{\text{regular}}}{r}+\frac{(2n+1)a}{2r(\ln r)^{2n+2}}-\frac{1}{2}\frac{df_{\text{regular}}}{dr}\right] \right\} \nonumber \\
	&=& \frac{1}{\sqrt{h_0}}\cdot\lim_{r\to 0}\frac{a(2n+1+2\ln r)}{2r(\ln r)^{2n+2}} 
	+\frac{1}{\sqrt{h_0}}\cdot\lim_{r\to 0}\frac{f_{\text{regular}}}{r} 
	-\frac{1}{2\sqrt{h_0}}\cdot\lim_{r\to 0}\frac{df_{\text{regular}}}{dr} \nonumber \\
	&=& +\infty	+\infty -\text{finite}	\nonumber \\
	&=& +\infty	.
\end{eqnarray}
It can be clearly observed that the two metric behaviors adopted in case II lead to the same conclusions. When naked singularity spacetime equipped with metric $f(r)=\frac{a}{(\ln r)^{2n}}+f_{\text{regular}}$ and $f(r)=\frac{a}{(\ln r)^{2n+1}}+f_{\text{regular}}$, the geodesic curvature  $\kappa_g(r)$ at the central limit satisfies $\lim_{r\to 0}\kappa_{g}(r)=+\infty$. Combined with the fact that the geodesic curvature satisfies $\kappa_{g}(r) \to 0^+$ or $\kappa_{g}(r)>0$ at the outer boundary (see Appendix \ref{appendix A}), the equation $\kappa_{g}(r)=0$ either has no solution or admits an even number of solutions. In such cases, there exists a total number of $N=2k$ circular photon orbits, with $k$ stable orbits and $k$ unstable circular orbits (according to the alternating distribution property of stable and unstable circular photon orbits in Appendix \ref{appendix B}). The topological invariant associated with circular photon orbits becomes $w=n_{\text{stable}} - n_{\text{unstable}} = 0$.

\textbf{Case III: Spacetime with $f(r)=a\left(\ln \frac{b}{r}\right)^{-\alpha}+f_{\text{regular}}$ behavior near the naked singularity.} In this part, we consider the metric given by $f(r)=a\left(\ln \frac{b}{r}\right)^{-\alpha}+f_{\text{regular}}$ with $0<\alpha<1$, $a>0$, $b>0$. Since it is required that $f(r)>0$ must be kept everywhere in the spacetime region, it follows that the regular part $f_{\text{regular}}$ should be positive at the naked singularity
\begin{equation}
	\lim_{r\to 0}f(r) 
	= \lim_{r\to 0}\left[a\left(\ln \frac{b}{r}\right)^{-\alpha}+f_{\text{regular}}\right] 
	\ \ \Rightarrow \ \ 
	\lim_{r\to 0}f_{\text{regular}} > 0 .
\end{equation}
Moreover, the first-order derivative of the metric diverges to positive infinity at the central limit
\begin{equation}
	\lim_{r\to 0}\frac{df(r)}{dr}
	= \lim_{r\to 0}\left[\frac{a\alpha}{r \left(\ln\frac{b}{r}\right)^{1+\alpha}} +\frac{df_{\text{regular}}}{dr}\right] 
	= +\infty +\text{finite} 
	= +\infty ,
\end{equation}
satisfying the classification of naked singularity spacetimes discussed in this subsection ($\lim_{r\to 0}\frac{df(r)}{dr}=+\infty$). Let $h_0=\lim_{r\to 0}[g(r) \cdot f(r)]$ denotes the product of metric functions at the naked singularity, then the geodesic curvature $\kappa_{g}(r)$ at the central limit becomes
\begin{eqnarray}
    \lim_{r\to 0}\kappa_{g}(r)
    &=& \lim_{r\to 0}\left[\frac{1}{\sqrt{f(r)\cdot g(r)}}\left(\frac{f(r)}{r}-\frac{1}{2}\frac{df(r)}{dr}\right)\right] \nonumber \\
    &=& \lim_{r\to 0} \left\{ \frac{1}{\sqrt{h_0}} \left[ \frac{a\left(\ln\frac{b}{r}\right)^{-\alpha}+f_{\text{regular}}}{r} -\frac{a\alpha}{2r\left(\ln\frac{b}{r}\right)^{1+\alpha}} -\frac{1}{2}\frac{df_{\text{regular}}}{dr} \right] \right\} \nonumber\\
    &=& \frac{1}{\sqrt{h_0}}\cdot\lim_{r\to 0}\frac{2a(\ln\frac{b}{r})-a\alpha}{2r \left(\ln\frac{b}{r}\right)^{1+\alpha}} 
    +\frac{1}{\sqrt{h_0}}\cdot\lim_{r\to 0}\frac{f_{\text{regular}}}{r} 
    -\frac{1}{2\sqrt{h_0}}\cdot\lim_{r\to 0}\frac{df_{\text{regular}}}{dr} \nonumber \\
    &=& +\infty +\infty - \text{finite} \nonumber \\
    &=& +\infty
\end{eqnarray}
It can be observed from the above choice of metric component $f(r)=a(\ln \frac{b}{r})^{-\alpha}+f_{\text{regular}}$ that the geodesic curvature diverges to positive infinity in the central limit ($\lim_{r\to 0}\kappa_{g}(r)=+\infty$). In addition, the geodesic curvature satisfies $\kappa_{g}(r) \to 0^+$ or $\kappa_{g}(r) >0$ at the outer boundary (see Appendix \ref{appendix A}). Subject to such conditions, the equation $\kappa_g(r)=0$ admits no roots or an even number of roots, which corresponds to $N=2k$ circular photon orbits in total. Furthermore, owing to the alternating distribution property of circular photon orbits, stable and unstable circular photon orbits are equal in number, with each type consisting of $k$ orbits. Consequently, the topological invariant associated with circular photon orbits becomes $w=n_{\text{stable}} - n_{\text{unstable}} = 0$.

\textbf{Summary and Conjecture.}
When examining the aforementioned representative cases, one can easily observe that the existence and number of circular photon orbits exhibit the same properties in these examples. For all three cases, in which the adopted metric functions satisfy $f(r)>0$ everywhere and their first derivative diverges to positive infinity $\lim_{r\to 0}\frac{df(r)}{dr}=+\infty$, the geodesic curvature tends to a positive value at the center, specifically $\lim_{r\to 0}\kappa_g(r)=+\infty$. Combined with the condition that the geodesic curvature at the outer boundary satisfies $\kappa_{g}(r) \to 0^+$ or $\kappa_{g}(r) > 0$, the geodesic constraint equation $\kappa_g(r)=0$ admits either no solution or an even number of solutions. The total number of circular photon orbits is even $N=2k$, comprising $k$ unstable and $k$ stable orbits. Therefore, it is reasonable to conjecture that the emergence of an even number of circular photon orbits may persist for a broader class of metric functions $f(r)$ (provided that the metric derivative satisfies $\lim_{r \to 0}\frac{df(r)}{dr}=+\infty$), yielding a general conclusion for naked singularity spacetimes. Accordingly, we propose the following conjecture:
\begin{conjecture} \label{conjecture 1}
For naked singularity spacetimes (in which the number of circular photon orbits is finite) and the metric derivative satisfies $\lim_{r\to 0}\frac{df(r)}{dr}=+\infty$ at the singularity point, then the total number of circular photon orbits is even $N=2k$ (where $k\in \mathbb{N}$), and the number of stable and unstable circular orbits are both equal to $k$.
\end{conjecture}

\subsection{Proof of the Conjecture \ref{conjecture 1} in Subsection \ref{sec3.2}}\label{sec3.3}

Following our geometric method for circular photon orbits, the validity of the conjecture \ref{conjecture 1} imposes specific requirements on the behavior of the geodesic curvature near the naked singularity. Given the outer boundary constraint that the geodesic curvature obeys $\kappa_{g}(r) \to 0^+$ or $\kappa_{g}(r) > 0$ for asymptotically flat, asymptotically de Sitter and anti-de Sitter spacetimes, it is required that the geodesic curvature $\kappa_{g}(r)>0$ must hold within a neighborhood of the naked singularity $r=0$. Based on this observation, we state the following proposition and provide its rigorous mathematical proof.

\begin{proposition} \label{propositin 1}
Suppose that the equation $\kappa_{g}(r)=0$ has only finitely many solutions for $r>0$, and that the metric function $f(r)$ of a naked singularity spacetime is smooth and positive for $r>0$, with its derivative satisfying $\lim_{r\to 0}\frac{df(r)}{dr}=+\infty$. Then there exists a sufficiently small quantity $\delta>0$ such that $\kappa_{g}(r)>0$ holds for all $r$ in the neighborhood of the center ($0 < r < \delta$).
\end{proposition}
\begin{proof}
We prove this proposition by contradiction. Suppose that for any $\delta>0$, there exists some $r\in (0,\delta)$ such that the metric function $f(r)$ satisfies $\kappa_{g}(r)<0$. Since we have assumed that $\kappa_{g}(r)=0$ admits finitely many solutions for $r>0$, there exists a smallest solution (denoted by $r_{\text{min}}$) such that the sign of the geodesic curvature is kept to be unchanged in the interval $r\in(0,r_{\text{min}})$ \footnote{It is worth noting that the geodesic curvature, calculated using the first derivatives of the metric functions, is always continuous throughout the spacetime region excluding the singularity (i.e., for $r>0$).}. By the contradiction hypothesis, when we choose a sufficiently small quantity $\delta\in(0, r_{\text{min}})$, there exists some $r\in(0,\delta)$ such that $\kappa_{g}(r)<0$. Since the sign of the geodesic curvature cannot change for $0<r<\delta<r_{\text{min}}$ (i.e., either $\kappa_{g}(r)<0$ or $\kappa_{g}(r)>0$ throughout this interval) \footnote{On the other hand, if $\kappa_{g}(r)=0$ admits no solution (so $r_{\text{min}}$ is not defined in this case), then for every $\delta>0$, the sign of $\kappa_{g}(r)$ is kept unchanged on $(0,\delta)$, and the subsequent results remain valid.}, it is easily verified that the existence of some $r\in(0,\delta)$ with $\kappa_{g}(r)<0$ would inevitably imply that $\kappa_{g}(r)<0$ holds for every $r\in(0,\delta)$. Hence, in this case, we find a sufficiently small $\delta$ such that $\kappa_{g}(r)<0$ holds throughout the neighborhood $r\in(0,\delta)$.

From the expression for $\kappa_{g}(r)$ given in \eqref{eq:12}, it is clearly that $\kappa_{g}(r)<0$ implies
\begin{equation}
\frac{f(r)}{r}-\frac{1}{2}\frac{df(r)}{dr}<0 
\ \ 
\Rightarrow
\ \
\frac{f(r)}{r}<\frac{1}{2}\frac{df(r)}{dr}<\frac{df(r)}{dr} 
\ \
\Rightarrow
\ \
r\frac{df(r)}{dr}-f(r)>0.
\end{equation}
Now define an auxiliary function $m(r)=\frac{f(r)}{r}$ on $r \in (0,\delta)$. A direct calculation gives 
\begin{equation}
\frac{dm(r)}{dr}=\frac{r\frac{df(r)}{dr}-f(r)}{r^2}>0 \label{eq:4.3}
\end{equation}
For the naked-singularity spacetimes considered here (which contain no event horizon), we always assume that $f(r) > 0$ is kept for $r>0$, so it follows that $0$ is a lower bound for $m(r)$. Choose a fixed point $r_0=\frac{\delta}{2}$, since $m(r)$ is non-decreasing (from equation (\ref{eq:4.3})) on the interval $(0,\delta)$, we obtain
\begin{equation}
0 < m(r) \le m(r_0)
\end{equation}
for any point $r$ satisfying $0<r<r_0$.
The Monotone Bounded Convergence Theorem for functions \cite{tao2022analysis,mathematical2019} guarantees that the limit $L=\lim_{r\to 0^+}m(r)$ exists and satisfies
\begin{equation}
0 \le L \le m(r_0) < +\infty
\end{equation}
Hence, $L$ is a finite non-negative real number.

For any sufficiently small $r$ such that $2r\in (0,\delta)$, the function $f(r)$ is continuous on the closed interval $[r,2r]$ and differentiable on the open interval $(r,2r)$. By the Lagrange Mean Value Theorem\cite{tao2022analysis,mathematical2019}, there exists at least one point $\eta_r \in (r,2r)$ depending on $r$, such that
\begin{equation}
\frac{df(\eta_r)}{dr}=\frac{f(2r)-f(r)}{2r-r}=\frac{f(2r)-f(r)}{r}=\frac{2r \cdot m(2r)-r \cdot m(r)}{r}=2m(2r)-m(r). \label{eq:4.1}
\end{equation}
In simplifying this expression, we have applied $f(r)=r \cdot m(r)$. Since $\eta_{r} \in (r,2r)$, it follows that $\eta_{r} \to 0^+$ as $r$ taking the limit $r \to 0^+$. Taking the limit $r \to 0^{+}$ on both sides of equation (\ref{eq:4.1}), we obtain
\begin{equation}
\lim_{r \to 0^+}\frac{df(\eta_r)}{dr}=\lim_{r \to 0^+}[2m(2r)-m(r)] \label{eq:4.2}
\end{equation}
Since we have obtained $\lim_{r\to 0^+}m(r)=L$ from the Monotone Bounded Convergence Theorem, it follows that $\lim_{r \to 0^+} m(2r) = L$. Substituting these limits into the equation (\ref{eq:4.2}) yields
\begin{equation}
\lim_{r \to 0^+}\frac{df(\eta_r)}{dr}=2L-L=L.
\end{equation}
Then there exists a sequence of points $\{\eta_r \}$ converging to $0$ such that the corresponding derivative values $\frac{df(\eta_r)}{dr}$ tend to a finite constant $L$. However, the assumption of proposition \ref{propositin 1} requires $\lim_{r \to 0^+}\frac{df(r)}{dr}=+\infty$, implying $\lim_{r \to 0}\frac{df(\eta_r)}{dr} = +\infty$ along every sequence ${\eta_r}$ converging to $0$. This leads to a contradiction. Consequently, the above assumption $(\kappa_{g}(r)<0)$ is false, and the proposition \ref{propositin 1} has been proved.
\end{proof}

Combining proposition \ref{propositin 1} with the constraint imposed on the geodesic curvature at the outer boundary ($\kappa_{g}(r)\to 0^+ $ or $\kappa_{g}(r)>0$), we eventually demonstrate that the geodesic curvature equation $\kappa_{g}(r)=0$ admits either no solution or an even number of solutions. This completes the proof of conjecture \ref{conjecture 1} presented in Subsection \ref{sec3.2}.

\section{Comparison of Circular Photon Orbits in Naked Singularity Spacetimes and Other Categories of Spacetimes} \label{sec Comparison}

\begin{figure*}[t]
	\centering
	\includegraphics[width=0.775\textwidth]{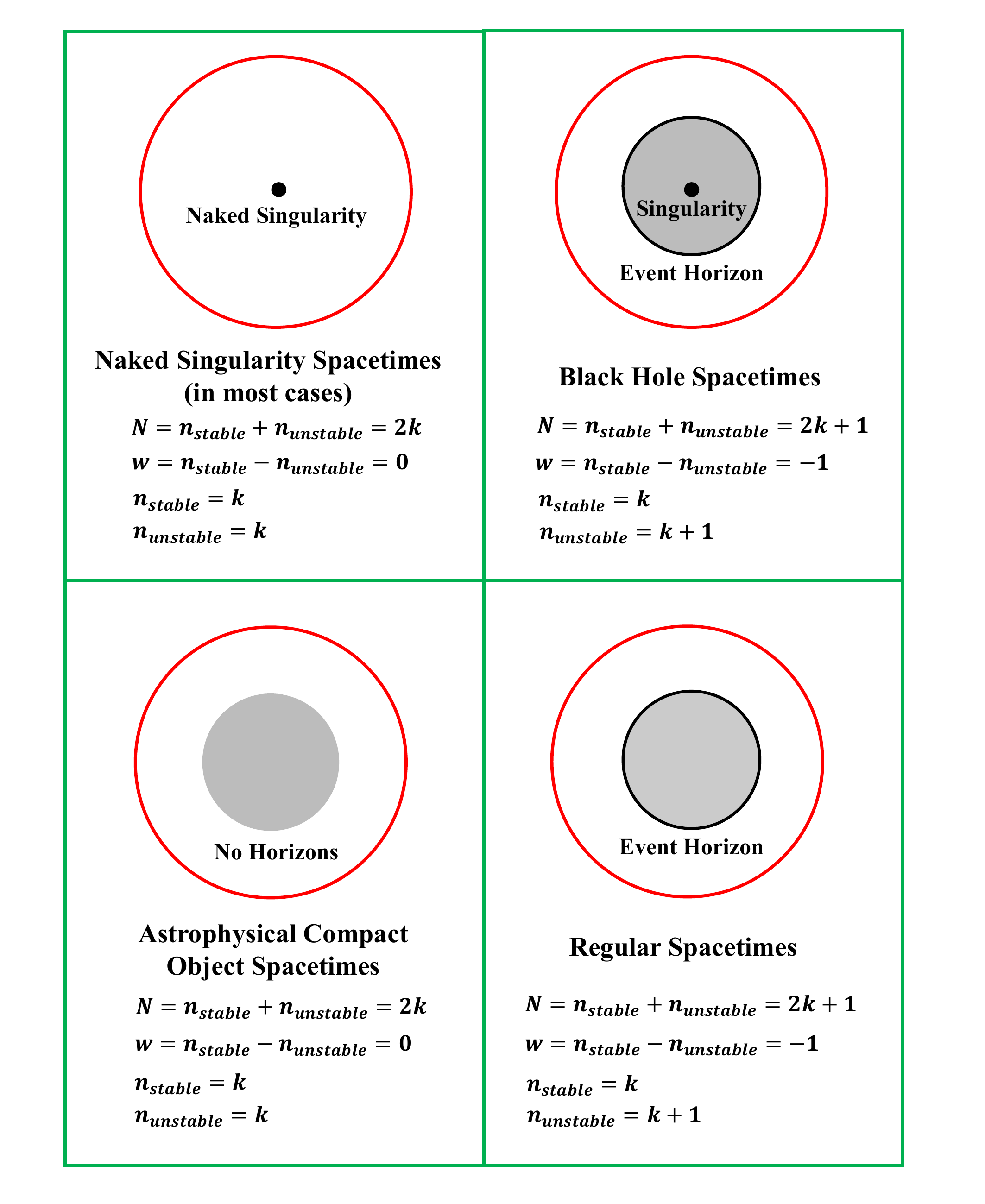}
	\caption{Comparison among the numbers of circular photon orbits in naked singularity spacetimes, black hole spacetimes, astrophysical compact object spacetimes, and regular spacetimes.}
	\label{fig:B}
\end{figure*}

In this work, we mainly focus on investigating circular photon orbits in naked singularity spacetimes. On the other hand, the properties of circular photon orbits in other categories of spacetimes (such as black hole spacetimes, astrophysical compact object spacetimes, regular spacetimes) have been extensively explored in recent studies. Therefore, it is highly valuable to provide a comparison between the features of circular photon orbits in different categories of spacetimes, which would enable us to identify the most dominant factors governing circular photon orbits in the gravitational field. Firstly, in black hole spacetimes possessing both singularities and event horizons, the total number of circular photon orbits is odd $N=2k+1$. Among these orbits, $k$ are stable circular photon orbits, and the remaining $k+1$ are unstable ones, so the spacetime topological invariant associated with circular photon orbits satisfies $w=-1$ \cite{Cunha_2020,cunha2022null,Wei_2020,Qiao_2025a}. Secondly, in singularity-free and horizonless spacetimes produced by compact objects, the total number of circular photon orbits is even $N=2k$, comprising $k$ stable circular photon orbits and $k$ unstable ones. In this case, the spacetime topological invariant satisfies $w=0$ \cite{Cunha_2017,Qiao_2025b}. Thirdly, in regular spacetimes with event horizons but no singularities, the total number of circular photon orbits is an odd integer $N=2k+1$. There are $k$ stable circular photon orbits and $k+1$ unstable ones, and the spacetime topological invariant is $w=-1$ \cite{Qiao_2025b}. Finally, the discussions and derivations presented in Section \ref{sec PO Naked Singularity} suggest that for naked singularity spacetimes containing singularities but devoid of event horizons, the total number of circular photon orbits is an even integer $N=2k$. Specifically, there exist $k$ stable circular photon orbits and $k$ unstable circular photon orbits, and the spacetime topological invariant satisfies $w=0$. These are the common features of circular photon orbits in naked singularity spacetimes. The only exceptional case arises when $\lim_{r\to0}f(r)=0$ and $\lim_{r\to0}\frac{df}{dr}=0$ are satisfied simultaneously, as discussed in a subcase given in Appendix \ref{appendix C}.

Comparing the circular photon orbits in the naked singularity spacetimes with those in black hole spacetimes, astrophysical compact object spacetimes, and regular spacetimes reported in the literature, we can draw the following conclusions. For black-hole spacetimes and regular spacetimes with event horizons, the total number of circular photon orbits is odd ($N=2k+1$). For a fixed horizon structure, the presence or absence of internal singularities inside the event horizon does not influence the number of circular photon orbits. In contrast, the horizonless spacetimes, including naked singularity and compact object spacetimes, correspond to an even number of circular photon orbits ($N=2k$). For horizonless geometries, the existence or nonexistence of naked singularities does not appear to play a decisive role in the existence or the number of circular photon orbits. This comparison implies that the properties of circular photon geodesics are primarily governed by the presence of event horizons rather than spacetime singularities.  

The fact that the event horizon has a dominant impact on the number of circular photon orbits can be explained by the following reason. In a Lorentzian spacetime, event horizons are completely determined by the spacetime causal structure, which is a geometric and topological structure uniquely prescribed by light cones and photon trajectories within the spacetime. Accordingly, the presence or absence of event horizons dictates the characteristic features of circular photon orbits in the spacetime. By contrast, the  presence or absence of spacetime singularities is not uniquely determined by the causal structure alone, which also depends on additional requirements such as energy conditions and trapped surfaces \cite{hawking2023large,hawking1970,senoilla1998}. For this reason, spacetime singularities exert a weaker influence on circular photon orbits relative to event horizons.

\section{Conclusion and Prospects}\label{sec Conclusion}

The investigation of circular photon orbits is of great significance for gravitational theories, mathematical physics and astronomical observations. Naked singularity spacetimes constitute one of the exotic classes of spacetimes in general relativity and other alternative theories of gravity, which may have a profound influence on photon orbits or causal structures. 

This work provides a comprehensive investigation into the features of circular photon orbits in naked singularity spacetimes with spherical symmetry. Specifically, our study mainly focuses on the existence of circular photon orbits, the total number of these circular orbits, the numbers of stable and unstable circular orbits. A geometric approach is employed in the current work, in which the framework of optical geometry together with its intrinsic geodesic curvature and Gaussian curvature is adopted in the study of circular photon orbits. In particular, the existence of circular photon orbits is constrained via the geodesic curvature condition $\kappa_{g}(r)=0$, which is finally determined by analyzing the behavior of the geodesic curvature at the naked singularity and at the outer boundary. Under the condition of $\lim_{r\to0}f(r)>0$, we examine various classes of naked singularity spacetimes with different limiting behaviors of the metric derivative ($\lim_{r\to 0}\frac{df(r)}{dr} =\text{finite}$, $\lim_{r\to 0}\frac{df(r)}{dr}=-\infty$, $\lim_{r\to 0}\frac{df(r)}{dr}=+\infty$) and summarize the common features of circular photon orbits in naked singularity spacetimes. We find that the total number of circular photon orbits is even ($N = 2k$), comprising of $k$ stable orbits and $k$ unstable orbits in equal proportion. A rigorous mathematical proof of this conclusion is also provided in our work. In addition, combining the analysis in Appendix \ref{appendix C} for $\lim_{r\to0}f(r)=0$, it is shown that this conclusion on circular photon orbits holds for the most of naked singularity spacetimes, with the only exceptional case arising when $\lim_{r\to0}f(r)=0$ and $\lim_{r\to0}\frac{df}{dr}=0$ are satisfied simultaneously (as a subcase in Appendix \ref{appendix C}). Furthermore, comparing our conclusions for circular photon orbits in naked singularity spacetimes with those in black hole spacetimes, regular spacetimes, compact celestial objects' spacetimes, we extract the influences of the spacetime singularity and the event horizon on circular photon orbits. The results indicate that the event horizon exerts stronger influences on the number of circular photon orbits in the gravitational field, compared with the spacetime singularity. 

The present work carries out a comprehensive study on circular orbits of massless photons. Nevertheless, similar analyses can be extended to investigate other families of circular particle orbits in naked singularity spacetimes. In particular, the effect of the naked singularity on the innermost stable circular orbit (ISCO) is also critical and deserves an in-depth study. Such an analysis can be performed when the intrinsic curvatures in the Jacobi reference geometry are utilized.

\begin{appendices}

\section{Behavior of the Geodesic Curvature at the Outer Boundary
}\label{appendix A}

In this Appendix, we present the analysis on the behavior of geodesic curvature in optical geometry at the outer boundary. It can be observed in the following part that the geodesic curvature at the outer boundary is independent of the existence of event horizons or spacetime singularities, and it is solely determined by the asymptotic behavior of the spacetime metric. In this work, we conduct the analysis for three most-common asymptotic spacetime configurations, which are asymptotically flat, asymptotically de Sitter, and asymptotically anti-de Sitter spacetimes. For asymptotically flat and asymptotically anti-de Sitter spacetimes, the outer boundary is located at the infinity ($r\to \infty$); for asymptotically de Sitter spacetime, the outer boundary is the cosmological horizon ($r\to r_{\text{cosmic}}$).

\textbf{Asymptotically Flat Spacetimes:} As the radial coordinate approaches the infinity $r\to \infty$, the spacetime metric of the asymptotically flat spacetime has the following asymptotic expansion.
\begin{subequations}
\begin{eqnarray}
	\lim_{r\to \infty}f(r)&=&1+\frac{a_i}{r^i}+O\left(\frac{1}{r^{i+1}}\right)\qquad(i\ge 1)\\
	\lim_{r\to \infty}g(r)&=&\frac{1}{1+\frac{b_j}{r^j}+O\left(\frac{1}{r^{j+1}}\right)}\qquad(j\ge 1)		
\end{eqnarray}
\end{subequations}
where $a_i$ and $b_j$ are expansion coefficients. Utilizing the above asymptotic expansion of the spacetime metric at infinity, the geodesic curvature $\kappa_{g}(r)$ at the outer boundary is calculated as:
\begin{eqnarray}
	\lim_{r\to \infty}\kappa_g(r)&=&\lim_{r\to \infty}\left[\frac{1}{\sqrt{f(r)\cdot g(r)}}\left(\frac{f(r)}{r}-\frac{1}{2}\frac{df(r)}{dr}\right)\right]\notag\\
	&=&\lim_{r\to \infty}\frac{\frac{1}{r}+\frac{a_i}{r^{i+1}}+\frac{ia_i}{2r^{i+1}}+O\left(\frac{1}{r^{i+2}}\right)}{\sqrt{1+\frac{\frac{a_i}{r^i}-\frac{b_j}{r^j}+O\left(\frac{1}{r^{i+1}},\frac{1}{r^{j+1}}\right)}{1+\frac{b_j}{r^j}+O\left(\frac{1}{r^{j+1}}\right)}}} \nonumber\\
	&=&\lim_{r\to \infty}\frac{1}{r} \nonumber\\
	&=&0^+			
\end{eqnarray}
The result indicates that the geodesic curvature tends to zero strictly from the positive side ($\kappa_{g}(r)>0$) as $r$ approaches infinity. This is a property of geodesic curvature that is universally held for naked singularity spacetimes with asymptotic flatness.

\textbf{Asymptotically Anti-de Sitter Spacetimes:} The spacetime metrics of asymptotically anti-de Sitter spacetimes admit the following asymptotic expansions in the limit of ${r\to \infty}$
\begin{subequations}
\begin{eqnarray}
	\lim_{r\to \infty}f(r)&=&1-\frac{\Lambda r^2}{3}+\frac{a_i}{r^i}+O\left(\frac{1}{r^{i+1}}\right)\qquad(i\ge 1)\\
	\lim_{r\to \infty}g(r)&=&\frac{1}{1-\frac{\Lambda r^2}{3}+\frac{b_j}{r^j}+O\left(\frac{1}{r^{j+1}}\right)}\qquad(j\ge 1)		
\end{eqnarray}
\end{subequations}
where $a_i$ and $b_j$ are expansion coefficients, and $\Lambda$ represents the cosmological constant. For asymptotically anti-de Sitter spacetimes, the cosmological constant satisfies $\Lambda<0$. Adopting the above asymptotic expansion of the spacetime metric at infinity, the behavior of the geodesic curvature $\kappa_{g}(r)$ at the outer boundary is expressed as:
\begin{eqnarray}
	\lim_{r\to \infty}\kappa_g(r)&=&\lim_{r\to \infty}\left[\frac{1}{\sqrt{f(r)\cdot g(r)}}\left(\frac{f(r)}{r}-\frac{1}{2}\frac{df(r)}{dr}\right)\right]\notag\\
	&=&\lim_{r\to \infty}\frac{\frac{1}{r}+\frac{a_i}{r^{i+1}}+\frac{ia_i}{2r^{i+1}}+O\left(\frac{1}{r^{i+2}}\right)}{\sqrt{1+\frac{\frac{a_i}{r^i}-\frac{b_j}{r^j}+O\left(\frac{1}{r^{i+1}},\frac{1}{r^{j+1}}\right)}{1-\frac{\Lambda r^2}{3}+\frac{b_j}{r^j}+O\left(\frac{1}{r^{j+1}}\right)}}} \nonumber\\
	&=&\lim_{r\to \infty}\frac{1}{r} \nonumber\\
	&=&0^+			
\end{eqnarray}
The results show that the geodesic curvature satisfies $\kappa_{g}(r)\to 0$ strictly from the positive side ($\kappa_{g}(r)>0$) at infinity, which is the same as the conclusion obtained for asymptotically flat spacetimes.
 
\textbf{Asymptotically de Sitter Spacetimes:} Due to the positive cosmological constant $\Lambda>0$, the asymptotically de Sitter spacetime has a cosmological horizon $r_\text{cosmic}$, such that the physical radial coordinate obeys $r \in (0,r_\text{cosmic})$. Furthermore, the metric function $f(r)$ decreases monotonically to zero as $r$ approaches the cosmological horizon (since the spacetime metric satisfies $f(r)=0$ at $r=r_\text{cosmic}$ and $f(r)>0$ when $r<r_\text{cosmic}$), which yields the corresponding results at the outer boundary 
\begin{eqnarray}
    &&f(r_\text{cosmic})=0, \nonumber
    \\
    &&\left. \frac{df(r)}{dr}\right|_{r=r_\text{cosmic}}<0,		
    \ \ \ \text{or} \ \ \
    \left. \frac{df(r)}{dr}\right|_{r=r_\text{cosmic}}=0^{-}.
\end{eqnarray}
Assuming that $f(r)\cdot g(r)$ does not vanish as the radial coordinate approaches the cosmological horizon $r\to r_{\text{cosmic}}$ \footnote{The requirement that $f(r)\cdot g(r)$ does not vanish indicates that the spacetime volume $dV = \sqrt{|\text{det}(g_{\mu\nu})|} d^4 x$ remains well-defined and non-degenerate as $r$ approaches the cosmological horizon $r\to r_{\text{cosmic}}$.}, when $\left. \frac{df(r)}{dr}\right|_{r=r_\text{cosmic}}<0$, the geodesic curvature at this outer boundary becomes 
\begin{equation}
    \lim_{r\to r_\text{cosmic}}\kappa_g(r)
    = \lim_{r\to r_\text{cosmic}}\left[\frac{1}{\sqrt{f(r)\cdot g(r)}}\left(\frac{f(r)}{r}-\frac{1}{2}\frac{df(r)}{dr}\right)\right]
    >0			
\end{equation}
On the other hand, when $\left. \frac{df(r)}{dr}\right|_{r=r_\text{cosmic}}=0^{-}$ as $r$ approaches the cosmological horizon (in the one-sided limit  $r<r_{\text{cosmic}}$), the geodesic curvature at the outer boundary gets 
\begin{eqnarray}
	\lim_{r\to r_\text{cosmic}}\kappa_g(r)
	&=& \lim_{r\to r_\text{cosmic}}\left[\frac{1}{\sqrt{f(r)\cdot g(r)}}\left(\frac{f(r)}{r}-\frac{1}{2}\frac{df(r)}{dr}\right)\right] \nonumber
	\\
	&=& \lim_{r\to r_\text{cosmic}}\left[\frac{1}{\sqrt{f(r)\cdot g(r)}}\right] \cdot \big(0^{+} - 0^{-}\big) \nonumber
	\\
	&=& 0^{+}			
\end{eqnarray}
The above results indicates that, in asymptotically de Sitter spacetimes, the geodesic curvature always get a non-negative value (namely, $\kappa_{g}(r)>0$ or $\kappa_{g}(r)=0^{+}$) at the outer boundary $r\to r_{\text{cosmic}}$.
 
Combining the discussions in this appendix, we can draw the following conclusion. For the most general asymptotic behavior of the Lorentzian spacetimes (asymptotically flat, asymptotically de Sitter, asymptotically anti-de Sitter), the geodesic curvature at the outer boundary exhibits $\kappa_{g}(r) \to 0^{+}$ or $\kappa_{g}(r)>0$.

\section{Alternating Distribution Property for Circular Photon Orbits} \label{appendix B}

\begin{figure*}
	\centering
	\includegraphics[width=0.85\textwidth]{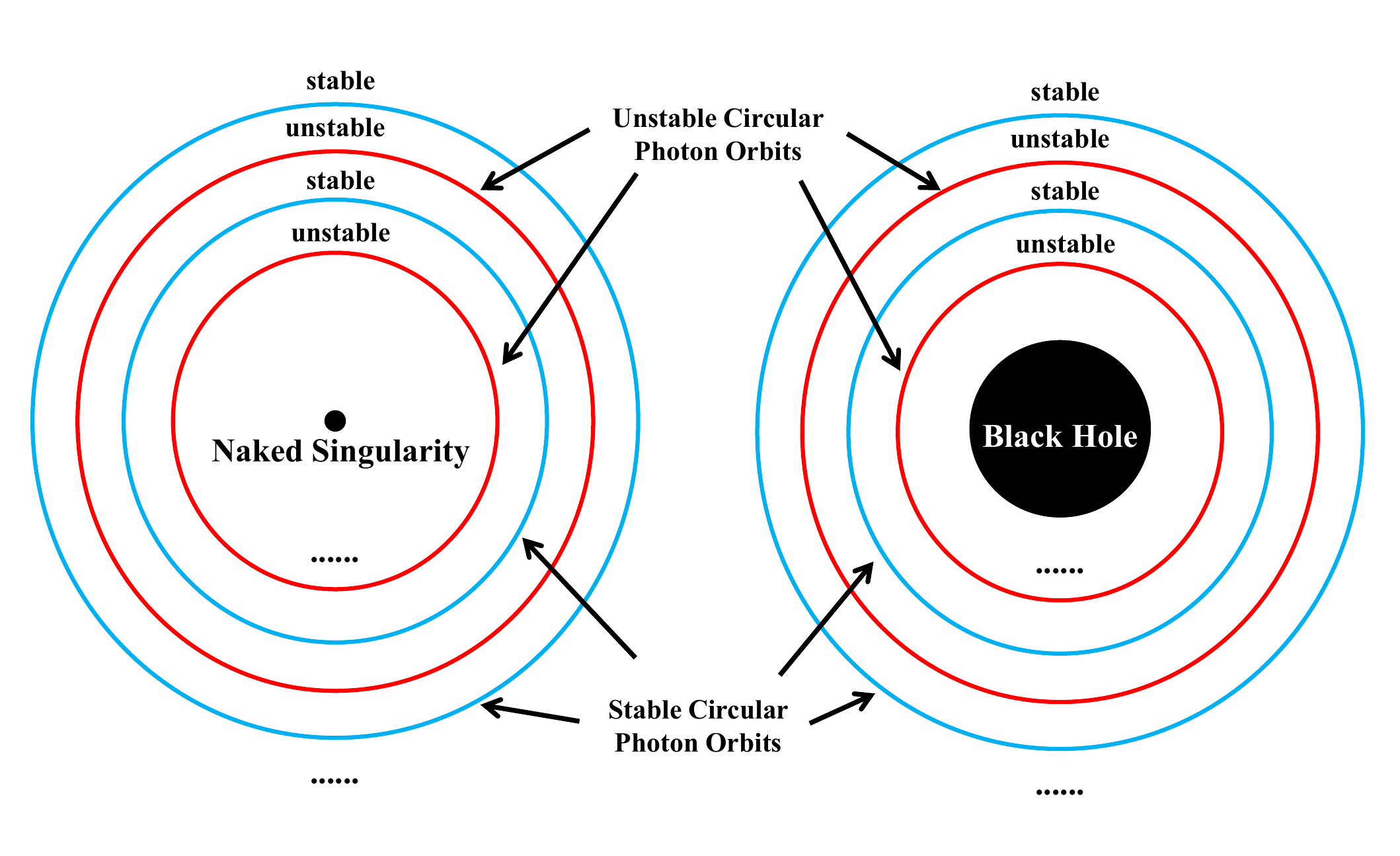}
	\caption{Illustration of the alternating distribution for stable and unstable circular photon orbits in black hole spacetimes and naked singularity spacetimes.}
	\label{fig:C}
\end{figure*}

In the gravitational field, the number of stable and unstable circular photon orbits is intrinsically linked to their underlying distribution property. Specifically, the stable and unstable circular photon orbits in spherically symmetric spacetimes must satisfy the alternating distribution property summarized below
\begin{quote} 
\textbf {Alternating Distribution Property of Circular Photon Orbits}: Stable and unstable circular photon orbits follow a strict alternating distribution. Specifically, there exists precisely one unstable circular photon orbit sandwiched between two adjacent stable circular photon orbits; likewise, exactly one stable circular photon orbit lies between any two neighbouring unstable circular photon orbits.
\end{quote}
This property for circular photon orbits can be rigorously derived mathematically. Particularly, reference \cite{Qiao_2025a} provides a derivation based on the Gauss-Bonnet theorem, which was originally performed for black hole spacetimes. However, one can see that the detailed procedure of derivation given in reference \cite{Qiao_2025a} does not depend on the assumptions concerning spacetime singularities or event horizons. Consequently, this alternating distribution property also holds for naked singularity spacetimes investigated in the present work (see the illustration in figure \ref{fig:C}).

\section{Behavior of the Geodesic Curvature at the Naked Singularity Under the Condition of $\lim_{r\to 0}f(r)=0^{+}$} \label{appendix C}

In Section \ref{sec PO Naked Singularity}, the behavior of the geodesic curvature at the inner boundary was examined under the implicit assumption that the metric function $f(r)$ remains non-vanishing at the central singularity, i.e., $\lim_{r\to 0}f(r)>0$. In that condition, the ratio $\frac{f(r)}{r}$ appearing in the brackets of the geodesic curvature expression \eqref{eq:12} diverges to positive infinity. In this Appendix, we investigate the complementary regime, in which the metric function vanishes at the naked singularity, namely
\begin{equation}
	\lim_{r\to 0}f(r)=0^{+}, \qquad f(r)>0 \ \ \text{for all } r>0 .
	\label{eq:C:setup}
\end{equation}
Under this condition, the limit of $\frac{f(r)}{r}$ can no longer be determined directly, and the behavior of the geodesic curvature at the inner boundary becomes subtle. In parallel with Section \ref{sec PO Naked Singularity} of the main text, we organize the discussion of the geodesic curvature behavior in the limit $r \to 0$ according to the three different cases of metric derivative (which are $\lim_{r\to 0}\frac{df(r)}{dr}=-\infty$, $\lim_{r\to 0}\frac{df(r)}{dr}=+\infty$, $\lim_{r\to 0}\frac{df(r)}{dr}=\text{finite}$), and finally summarize the conclusions on the number of circular photon orbits.

\textbf{Case I: $\mathbf{\lim_{r\to 0}\frac{df(r)}{dr}=-\infty}$.} Unlike in Section \ref{sec PO Naked Singularity}, this case cannot be realized under the imposition of \eqref{eq:C:setup}. The reason for the exclusion of this case can be naturally explained from the negative divergent behavior of the metric derivative. Evidently, $\lim_{r\to 0}\frac{df(r)}{dr}=-\infty$ implies that, for any fixed $M>0$, there exists $\delta>0$ such that $\frac{df(r)}{dr}<-M$ holds for $r \in (0,\delta)$. Integrating both sides from $s$ to $r$ ($0<s<r<\delta$) gives 
\begin{equation}
f(r)-f(s)=\int_{s}^{r}\frac{df(t)}{dt}dt<-M(r-s). 
\end{equation}
Taking $s\to 0^{+}$ and using $f(s)\to 0$, one obtains a negative function with $f(r)\le -Mr<0$ on the neighborhood $(0,\delta)$, which is unphysical for naked singularity spacetimes and it contradicts the requirement $f(r)>0$ for $r>0$. Hence the case $\lim_{r\to 0}\frac{df(r)}{dr}=-\infty$ cannot be realized when $f(r)\to 0^{+}$.

\textbf{Case II: $\mathbf{\lim_{r\to 0}\frac{df(r)}{dr}=+\infty}$.} The analysis of circular photon orbits in this case can be proceeded in parallel with that given in subsections \ref{sec3.2} and \ref{sec3.3}. First, we begin with three representative examples of naked singularity spacetime metrics $f(r)$ discussed in subsection \ref{sec3.2}. To meet the condition of $\lim_{r \to 0}f(r)=0$, the regular terms in these metrics must satisfy $\lim_{r \to 0}f_{\text{regular}}=0$. Specifically, for the first example $f(r)=ar^{\alpha}+f_{\text{regular}}$, we have
\begin{equation}
	\lim_{r\to 0}f(r) 
	=\lim_{r\to 0}\left[ar^{\alpha}+f_{\text{regular}}\right] = 0 
	\ \ \Rightarrow \ \ 
	\lim_{r\to 0}f_{\text{regular}} = 0 .
\end{equation}
and the geodesic curvature in the limit $r\to 0$ becomes
\begin{eqnarray}
	\lim_{r\to 0}\kappa_{g}(r)
	&=& \frac{1}{\sqrt{h_{0}}}\cdot\lim_{r\to 0}\frac{2a-a\alpha}{2r^{1-\alpha}} +\frac{1}{\sqrt{h_0}}\cdot\lim_{r\to 0}\frac{f_{\text{regular}}}{r} -\frac{1}{2\sqrt{h_0}}\cdot\lim_{r\to 0}\frac{df_{\text{regular}}}{dr} \nonumber \\
	&=& \frac{1}{\sqrt{h_{0}}}\cdot\lim_{r\to 0}\frac{2a-a\alpha}{2r^{1-\alpha}} +\frac{1}{2\sqrt{h_0}}\cdot\lim_{r\to 0}\frac{f_{\text{regular}}}{r}  \nonumber \\
	&=& +\infty +\frac{D}{2\sqrt{h_0}} \nonumber \\
	&=& +\infty		
\end{eqnarray}
where we have used L' H\^{o}pital's rule $\lim_{r\to 0}\frac{f_{\text{regular}}}{r}=\lim_{r\to 0}\frac{df_{\text{regular}}}{dr}$ in the second equality. 

By analogy, it can be easily verified that similar calculations and conclusions hold for other examples in subsection \ref{sec3.2} ($f(r)=\frac{a}{(\ln r)^{n}}+f_{\text{regular}}$ and $f(r)=a\left(\ln \frac{b}{r}\right)^{-\alpha}+f_{\text{regular}}$). Combined with the positive outer boundary behavior ($\kappa_{g}>0$ or $\kappa_{g} \to 0^{+}$), this implies that, for $f(r)$ satisfying these forms, the geodesic curvature equation has an even number of solutions, and consequently the number of circular photon orbits is even. 

Next, based on the above results, we arrive at the same conjecture stated in subsection \ref{sec3.2} (conjecture \ref{conjecture 1}), whose validity is ensured by proposition \ref{propositin 1}. Furthermore, the detailed proof of this proposition requires only $f(r)>0$ for $r>0$ (it does not impose any constraint on the value of $\lim_{r \to 0} f(r)$), and therefore it remains valid under the present condition $f(r)\to 0^{+}$
Consequently, we conclude that this conjecture still holds for  $f(r)$ satisfying the requirement \eqref{eq:C:setup}. Once again, one obtains that circular photon orbits satisfy $N=2k$, with $k$ stable and $k$ unstable orbits, and the topological invariant becomes $w=0$.

\textbf{Case III: $\mathbf{\lim_{r\to 0}\frac{df(r)}{dr}}= \text{finite}$.}
When the metric derivative tends to a finite value at the center naked singularity, we can denote this value $C$ to be limiting value of metric derivative, namely $\lim_{r\to 0}\frac{df(r)}{dr} = C$. Applying the L' H\^{o}pital's rule in the limit of $r\to 0$ gives
\begin{equation}
	\lim_{r \to 0}\frac{f(r)}{r}=\lim_{r \to 0}\frac{df(r)}{dr}=C .
	\label{eq:C:mvt}
\end{equation}
In the following analysis, we classify the discussion into three subcases according to the positive and negative value of $C$.

\textbf{Case III(i) $\mathbf{\lim_{r\to 0}\frac{df(r)}{dr}=C<0}$.}
Similar to the Case I, this subcase should also be excluded by the condition of \eqref{eq:C:setup}. Since the metric of naked singularity spacetimes obeys $f(r)>0$ for $r>0$, one must have $C\ge 0$. Otherwise, a negative $C$ would imply $\frac{f(r)}{r} \to C<0$ at the naked singularity $r \to 0$ such that $f(r)<0$ appears for sufficiently small $r$, contradicting the requirement in equation \eqref{eq:C:setup}. 

\textbf{Case III(ii) $\mathbf{\lim_{r\to 0}\frac{df(r)}{dr}=C>0}$.} In this subcase, equation \eqref{eq:C:mvt} gives $\frac{f(r)}{r}\to C$, and the geodesic curvature at the inner boundary is calculated as
\begin{equation}
	\lim_{r\to 0}\kappa_{g}(r)
	= \lim_{r\to 0}\left[\frac{1}{\sqrt{f(r)\cdot g(r)}}\left(\frac{f(r)}{r}-\frac{1}{2}\frac{df(r)}{dr}\right)\right]
	=\frac{1}{\sqrt{h_0}}\left(C-\frac{C}{2}\right)
	>0.
	\label{eq:C:cpos}
\end{equation}
where we have used the fact that $f(r)\cdot g(r)\to h_0\in(0,+\infty)$ at the naked singularity. In this case, the geodesic curvature approaches a finite positive value $\lim_{r\to 0}\kappa_g(r)>0$. Together with the positive outer boundary behavior ($\kappa_{g}(r)>0$ or $\kappa_{g}(r) \to 0^{+}$), the equation $\kappa_g(r)=0$ admits either no solution or an even number of solutions. Consequently, the total number of circular photon orbits is $N=2k$, with $k$ stable and $k$ unstable orbits, and the topological invariant is $w=0$, in complete agreement with the analysis presented in subsection \ref{sec3.1}.

\textbf{Case III(iii) $\mathbf{\lim_{r\to 0}\frac{df(r)}{dr}=C=0}$.} In this subcase, equation \eqref{eq:C:mvt} implies $\frac{f(r)}{r}\to C=0$, i.e., $f(r)=o(r)$, and both terms in the parentheses of the geodesic curvature expression tend to zero, so that the sign of  $\frac{f(r)}{r}-\frac{1}{2}\frac{df(r)}{dr}$ is determined by the higher-order behavior. 

To analyze the geodesic curvature behavior, it is also convenient to introduce an auxiliary function
\begin{equation}
\rho(r)\equiv \frac{f(r)}{r^{2}}.
\label{eq:C:aux}
\end{equation}
A direct differentiation of (\ref{eq:C:aux}) yields the following exact relation
\begin{eqnarray}
	-\frac{r^{2}}{2}\frac{d\rho(r)}{dr}
	&=& \frac{f(r)}{r} - \frac{1}{2}\frac{df(r)}{dr}
	\label{eq:C:rho}
\end{eqnarray}
Consequently, the sign of $\kappa_g(r)$ and the positions of its zeros are determined solely by 
\begin{subequations}
\begin{eqnarray}
	\operatorname{sgn}\kappa_g(r)
	=\operatorname{sgn} \left[ \frac{f(r)}{r} - \frac{1}{2}\frac{df(r)}{dr} \right]
	=-\operatorname{sgn} \left[\frac{d\rho(r)}{dr}\right],\label{eq:C:A}
	\\
	\kappa_g(r)=0 
	\ \Leftrightarrow \ 
	\frac{f(r)}{r} - \frac{1}{2}\frac{df(r)}{dr} =0 
	\ \Leftrightarrow \ 
	\frac{d\rho(r)}{dr} = 0.
	\label{eq:C:B}
\end{eqnarray}
\end{subequations}
The prefactor $1/\sqrt{f(r)g(r)}$ in the geodesic curvature expression affects only the magnitude (or the possible divergence) of $\kappa_g(r)$ \footnote{Provided that $f(r)\cdot g(r)\to h_0\in(0,+\infty)$ at the singularity, it approaches the finite value $1/\sqrt{h_0}$; if $f(r)\cdot g(r)\to 0$ the divergence of $\kappa_g(r)$ is enhanced, while if $f(r)\cdot g(r)\to+\infty$ it is suppressed}. In no case does the prefactor alter the sign of $\kappa_g(r)$ or the number of its zeros. Equation \eqref{eq:C:A} shows that this sign is controlled by the monotonicity of $\rho(r)=\frac{f(r)}{r^{2}}$. When $\lim_{r\to 0}\frac{d\rho(r)}{dr}<0$, \eqref{eq:C:B} yields $\kappa_{g}(r)>0$ at the naked singularity, in which case the spacetime admits an even number of circular photon orbits; conversely, when $\lim_{r\to 0}\frac{d\rho(r)}{dr}>0$, one finds $\kappa_{g}(r)<0$, corresponding to an odd number of circular photon orbits. Consequently, in the subcase $\lim_{r\to 0}\frac{df(r)}{dr}=C=0$, the parity of the total number of circular photon orbits remains indeterminate, and additional physical input is required to place a constraint on the number of such orbits.

It is also illuminating to relate $\rho(r)=\frac{f(r)}{r^{2}}$ to the radial effective potential of null geodesics (up to a positive factor proportional to the squared angular momentum): the reduced radial geodesic equation for spherically symmetric spacetimes gives
\begin{equation}
    E^{2} = f(r) \cdot g(r) \cdot \left(\frac{dr}{d\tau}\right)^{2} + f(r) \cdot \left( \frac{L^{2}}{r^{2}} + \epsilon \right) ,
\end{equation}
with $\epsilon=0$ for massless photon orbits. It is easily shown that the circular photon orbits are precisely the stationary points of $\rho(r)$, in agreement with $\kappa_g=0\Leftrightarrow \frac{d\rho(r)}{dr}=0$. A circular photon orbit is stable when $\rho(r)$ attains a local minimum ($\frac{d^{2}\rho}{dr^{2}}>0$) and unstable when it attains a local maximum ($\frac{d^{2}\rho}{dr^{2}}<0$). Thus, the two regimes $\frac{d\rho(r)}{dr}>0$  and $\frac{d\rho(r)}{dr}<0$ at the naked singularity naturally distinguish the possible numbers of critical points of $\rho(r)$, thereby determining the parity of the number of circular photon orbits.

Having established how the limiting behavior of $\frac{df(r)}{dr}$ determines the sign of $\kappa_{g}(r)$ and the number of circular photon orbits, we then collect all possible cases into a systematic classification. The characteristics of circular photon orbits under the condition $f(r)\to 0^{+}$ are summarized in Table \ref{tab:C}, which summarizes the geodesic curvature and the number of circular photon orbits according to the limiting behavior of $\frac{df(r)}{dr}$ at the naked singularity.

\begin{center}
	\textbf{Table C1.} Behavior of the geodesic curvature at the naked singularity and the number of circular photon orbits classified by the limiting behavior of $\frac{df(r)}{dr}$ for $f(r)\to 0^{+}$.
	\footnote{It is worth noting that a representative example of this situation is the Janis-Newman-Winicour (JNW) spacetime
		\begin{equation}
			ds^{2}
			=-\left(1-\frac{b}{r}\right)^{\nu}dt^{2}
			+\left(1-\frac{b}{r}\right)^{-\nu}dr^{2}
			+r^{2}\left(1-\frac{b}{r}\right)^{1-\nu}
			d\Omega^{2},
			\qquad
			\nu=\frac{2M}{b}
			=\frac{M}{\sqrt{M^{2}+q^{2}}},
		\end{equation}
		For $0<\nu\leq1$, the spacetime possesses a naked curvature singularity at $r=b$. When $\frac{1}{2}< \nu \leq 1$, there exists a circular photon orbit in the JNW spacetime, whereas for $0 < \nu < \frac{1}{2}$, no circular photon orbit exists in this JNW spacetime. 
		
		To give a comparison of the circular photon orbits in JNW spacetime and the conclusion in table \ref{tab:C}, we can introduce the a real radius coordinate 
		\begin{equation}
			r'=r\left(1-\frac{b}{r}\right)^{\frac{1-\nu}{2}},
		\end{equation}
		such that the naked singularity locates at $r'=0$. After the coordinate transformation, the spacetime metric is transformed into the metric form of (\ref{eq:9})
		\begin{equation}
			ds^{2} =-f(r')dt^{2}+g(r')dr'^{2}+r'^{2}d\Omega^{2}.
		\end{equation}
		Near the naked singularity $r'=0$, the metric function satisfies
		\begin{equation}
			f(r')\sim
			\left(\frac{r'}{b}\right)^{p},
			\ \ \Rightarrow \ \ 
			\lim_{r' \to 0} f(r') = 0.
		\end{equation}
		with $p=\frac{2\nu}{1-\nu}$. This is in accordance with the requirement (\ref{eq:C:setup}) in this appendix. A direct calculation gives the derivative of the metric function near the singularity
		\begin{equation}
			\frac{df(r')}{dr'} \sim \frac{p}{b}\left(\frac{r'}{b}\right)^{p-1},
		\end{equation}
		so that the limiting behavior of $\frac{df(r')}{dr'}$ at the naked singularity is
		controlled by the value of $p$, and hence by the parameter $\nu$. For $0<\nu<1/3$, one has $0<p<1$, and the derivative diverges to $\lim_{r'\to 0}\frac{df(r')}{dr'}=+\infty$ at the naked singularity. The JNW spacetime in this parameter regime therefore falls into the second row of table \ref{tab:C}, which predicts an even number of circular photon orbits with vanishing topological invariant. It is completely in agreement with the circular photon orbit in JNW spacetime (when $0<\nu<1/3$, the number of circular photon orbit is $N=0$. For $\nu=1/3$ one has exactly $p=1$, the metric function behaves linearly as $f(r')\approx \frac{r'}{b}$, and the derivative approaches the finite positive value $\lim_{r'\to 0}\frac{df(r')}{dr'}=\frac{1}{b}$. This realizes the fourth row of
		table \ref{tab:C} (with $\lim_{r'\to 0}\frac{df(r')}{dr'}=C>0$), which again predicts an even number of orbits. The actual number of circular photon orbits in JNW spacetime (when $\nu=1/3$) is $N=0$, consistently with the prediction in table \ref{tab:C}. For $1/3<\nu<1$ one has $p>1$ and the derivative vanishes at the singularity, so that the spacetime belongs to the last row of table \ref{tab:C} (with $\lim_{r'\to 0}\frac{df(r')}{dr'}=C=0$), for which the number of circular photon orbits is inconclusive (the number of circular photon orbits can be either odd or even). In particular, the JNW spacetime admits no circular photon orbits when the parameter takes the value $\frac{1}{3}<\nu<\frac{1}{2}$, while an unstable circular photon orbit appears in JNW spacetime when $\frac{1}{2}<\nu\neq 1$. Both of these two cases are still consistent with our conclusion presented in table \ref{tab:C}, since the case in the last row dose not imposes a constraint on the parity of the number of circular photon orbits (either $N=1$ or $N=0$ can be appears in this subcase).}
	\\[4pt]
	\begin{tabular}{ccccc}
		\toprule
		Range of  $\lim_{r\to 0}\frac{df(r)}{dr}$ & $\kappa_g(r)$ at the inner boundary & Number of orbits & Topological invariant \\
		\midrule
		$-\infty$ & - & - & - \\
		$+\infty$ & $+\infty$ & $N=2k$ & $w=0$ \\
		\text{finite} ($C<0$) & - & - & - \\
	    \text{finite} ($C>0$) & $\frac{C}{2}>0$ & $N=2k$ & $w=0$ \\
		\text{finite} ($C=0$) & $\text{inconclusive}$ & $\text{inconclusive}$ & $\text{inconclusive}$ \\
		\bottomrule
	\end{tabular}
	\label{tab:C}
\end{center}

\end{appendices}

\section*{Acknowledgements}

The authors thank Jin-Sheng Du and Jun-Hua Wu for helpful discussions and comments on the manuscript. We are especially grateful to Ye Zhou for the valuable examples and discussions he provided after the first version of our manuscript was posted on arXiv. This work is supported by the Natural Science Foundation of Chongqing Municipality (Grant No. CSTB2022NSCQ-MSX0932), and the Scientific and Technological Research Program of Chongqing Municipal Education Commission (Grant No. KJQN202201126).

\nocite{*}
\bibliography{./ref}





\end{document}